\documentclass[a4paper,fleqn]{cas-sc}
\usepackage[utf8]{inputenc}
\usepackage[numbers]{natbib}
\usepackage{amsthm,dsfont}
\definecolor{Maroon}{RGB}{128,0,0}
\definecolor{MidnightBlue}{RGB}{25,25,112}
\hypersetup{hidelinks}
\usepackage[capitalise,noabbrev]{cleveref}
\newtheorem{theorem}{Theorem}

\newcommand{\Wone}{\ensuremath{\mathtt{W}[1]}\xspace}
\newcommand{\MCclique}{\textup{\textsc{Multicolored Clique}}\xspace}
\newcommand{\bigparagraph}[1]{\par\medskip\noindent\textbf{#1}}

\newenvironment{tightcenter}
 {\parskip=0pt\par\nopagebreak\centering}
 {\par\noindent\ignorespacesafterend}

\usepackage{tikz}
\usetikzlibrary{calc}
\usepackage{xargs}
\usepackage{xifthen}
\usepackage{framed}

\usepackage{ctable}

\newlength{\RoundedBoxWidth}
\newsavebox{\GrayRoundedBox}
\newenvironment{GrayBox}[1]%
   {\setlength{\RoundedBoxWidth}{\textwidth-10.5ex}
    \def\boxheading{#1}
    \begin{lrbox}{\GrayRoundedBox}
       \begin{minipage}{\RoundedBoxWidth}%
   }{%
       \end{minipage}
    \end{lrbox}%
    \begin{tightcenter}%
    \begin{tikzpicture}%
       \node(Text)[draw=black!90,fill=white,rounded corners,%
             inner sep=2ex,text width=\RoundedBoxWidth]%
             {\usebox{\GrayRoundedBox}};
        \coordinate(x) at (current bounding box.north west);
        \node [draw=white,rectangle,inner sep=3pt,anchor=north west,fill=white] 
        at ($(x)+(6pt,.75em)$) {\boxheading};
    \end{tikzpicture}
    \end{tightcenter}\vspace{0pt}%
    \ignorespacesafterend
}

\begin{document}
\let\WriteBookmarks\relax
\def\floatpagepagefraction{1}
\def\textpagefraction{.001}
\shorttitle{Parameterized Hardness of Mixed 2-Sided Orthant Depth}
\shortauthors{M. D\"oring and G. Tennigkeit}
\title[mode=title]{Parameterized Hardness of Mixed 2-Sided Orthant Depth}
\author[1]{Michelle D\"oring}[orcid=0000-0001-7737-3903]
\cormark[1]
\ead{michelle.doering@hpi.de}
\ead[url]{https://michelledoering.notion.site/}
\author[1]{Georg Tennigkeit}[orcid=0000-0003-0734-0684]
\ead{georg.tennigkeit@hpi.de}
\affiliation[1]{organization={Hasso Plattner Institute, University of Potsdam},
    city={Potsdam}, country={Germany}}
\cortext[1]{Corresponding author.}
\begin{abstract}
We consider the maximum-depth problem for mixed 2-sided orthants in $\mathbb R^d$: each region imposes one lower bound and one upper bound on distinct coordinates, and the task is to find a point contained in as many regions as possible.
We show that the corresponding decision problem is \Wone-hard when parameterized by the dimension.
Our reduction from \MCclique uses two coordinates per color class and only polynomially many orthants.
\end{abstract}


\maketitle

\section{Introduction}
In the general \textup{\textsc{Box Depth}} problem, we are given a set of $n$ axis-aligned boxes in $\mathbb{R}^d$ together with an integer~$\ell$, and the task is to decide whether there is a point $x\in\mathbb{R}^d$ that is contained in at least $\ell$ boxes. A \emph{2-sided orthant} is an axis-aligned region that is unbounded in $d-2$ dimensions and bounded from one side in each of two distinct dimensions; that is, it is of the form $\{x\in\mathbb{R}^d : (x_i\,{\mathord ?}\,a)\land(x_j\,{\mathord ?}\,b)\}$ for distinct $i,j\in[d]$, where each occurrence of ${\mathord ?}$ may independently be $\leq$ or $\geq$~\cite{chan_kleesmeasure_2013}. We consider the further restricted \textup{\textsc{Mixed 2-Sided Orthant Depth (M2SOD)}} problem, in which every region has one upper-bounded and one lower-bounded coordinate. Thus, every region is of the form $\{x\in\mathbb{R}^d : x_i\leq a,\ x_j\geq b\}$ for some distinct $i,j\in[d]$.

The maximum-depth problem is closely related to Klee’s measure problem, which asks for the volume of the union of axis-aligned boxes. Algorithms for the latter can also be adapted to compute maximum depth; for general boxes, the resulting running times have an exponent approximately $d/2$ for fixed dimension $d$ \cite{chan_kleesmeasure_2013}. The dependence on the dimension is supported by conditional lower bounds. In particular, Gorbachev and Künnemann \cite{gorbachev_combinatorialdesigns_2023} establish lower bounds for unweighted box depth under the 3-uniform hyperclique hypothesis.

Restricting the geometry of the input regions can substantially improve algorithms for related problems. In his work on Klee’s measure problem \cite{chan_kleesmeasure_2013}, Chan develops a simplification procedure for 2-sided orthants, while explicitly noting that the resulting techniques do not apply to maximum depth.
Two-sided orthants also play a central role in algorithms for largest empty boxes \cite{chan_fasteralgorithms_2023}. These results motivate the question whether restricting each region to only two coordinate bounds makes maximum depth fixed-parameter tractable with respect to the dimension. With only one bound per region, the problem is polynomial-time solvable: the depth separates into independent one-dimensional functions, each of which can be maximized by sorting and scanning its thresholds.

We show that allowing two bounds per region already makes \textsc{Box Depth} W[1]-hard parameterized by the dimension, even when every region has exactly one lower bound and one upper bound on distinct coordinates and all bounds are integers. Our reduction from Multicolored Clique uses two coordinates per color class and polynomially many explicitly listed orthants with integer bounds.

\section{Result}
Formally, an instance of \textup{\textsc{Mixed 2-Sided Ortant Depth (M2SOD)}} consists of a dimension $d\geq2$, an explicitly listed indexed family $\mathcal O=(O_1,\ldots,O_N)$ of mixed 2-sided orthants with rational bounds encoded in binary, and a nonnegative integer threshold $\ell$ encoded in binary. The question is whether there exists $z\in\mathbb R^d$ such that $\operatorname{depth}_{\mathcal O}(z):=\bigl|\{r\in[N]:z\in O_r\}\bigr|\geq\ell$. 

Note that we allow repeated regions: the input is an indexed family of orthants, and depth counts each occurrence separately. The multiplicities used in our reduction are implemented by explicitly listing polynomially many copies. An equivalent construction with unique orthants would be possible by slightly offsetting the bound values in duplicate orthants; with integer bounds, all values can be scaled up accordingly to make space for such small offsets.

We write $[n]=\{1,\ldots,n\}$. Throughout, all bounds are nonstrict and all regions are closed. For ease of reading, we give names to dimensions instead of enumerating them with integers, and denote dimension $d$ of a point $z$ as $z[d]$.

\bigparagraph{Intuition.} We reduce from \textup{\textsc{Multicolored Clique}}.
For each color class $i$, we use two dimensions $x_i,y_i$ to encode the choice of a vertex in this class. Specifically, $z[x_i]=s=-z[y_i]$ will indicate choosing the vertex with index $s$ in color class $i$. We will call every point that satisfies $z[x_i]=s=-z[y_i]$ for every dimension $i$ a \textit{valid choice}.
The two coordinates let us express a bound on the selected index using either coordinate: for a point $z$ representing a valid choice, $z[x_i]\leq a$ is equivalent to $z[y_i]\geq-a$, and $z[x_i]\geq a$ is equivalent to $z[y_i]\leq-a$.
To ensure that every solution is a valid choice we introduce \textit{choice orthants} with sufficiently large multiplicity. 
For each edge of the input graph, four further orthants contribute a baseline of one at every valid choice, and a contribution of two precisely when the two endpoints of that edge are selected.
Reaching the target depth will then require an edge between every pair of selected vertices.
The construction uses only $2k$ dimensions for $k$ color classes, and every orthant has one lower and one upper bound on distinct coordinates.

\begin{theorem}\label{thm:M2SOD-wone-hard}
    \textup{\textsc{Mixed 2-Sided Orthant Depth}} is \Wone-hard when parameterized by the dimension.
\end{theorem}
\begin{proof}
    We give a parameterized reduction from \MCclique, which is \Wone-hard when parameterized by the number of color classes~\cite{fellows_parameterizedcomplexity_2009}. Here, the task is to choose one vertex from each color class so that the selected vertices form a clique. 

    Let $(H=(V_H,E_H), (V_1,\dots,V_k))$ be an instance of \MCclique for a simple undirected graph $H$. 
    Instances with $k<2$ can be decided directly, so we assume $k\geq2$. We additionally assume that every color class has the same size $n\geq1$, and  add isolated vertices where necessary.
    Write $V_i=\{v_{i,1},\ldots,v_{i,n}\}$ and, for $i<j$, let $E_{ij}$ denote the set of edges between $V_i$ and $V_j$.

    We construct an instance $(\mathcal O,\ell)$ of \textup{\textsc{M2SOD}} in dimension $d=2k$. For each $i\in[k]$, introduce two dimensions $x_i$ and $y_i$. A point $z\in\mathbb{R}^{2k}$ is a \emph{valid choice} if for every $i\in[k]$ there is an $\alpha_i\in[n]$ with $z[x_i]=\alpha_i$ and $z[y_i]=-\alpha_i$. Thus, a valid choice encodes the selection of one unique vertex $v_{i,\alpha_i}$ from each color class.

    The family of orthants $\mathcal O$ is composed as follows. For every vertex $v_{i,\alpha}$ with $i\in[k]$ and $\alpha\in[n]$, introduce two mixed 2-sided \textit{choice orthants}:\begin{align*}
        U_{i,\alpha}=\{z\in\mathbb{R}^{2k}:z[x_i]\geq \alpha,\ z[y_i]\leq -\alpha\},
        \quad\quad
        L_{i,\alpha}=\{z\in\mathbb{R}^{2k}:z[x_i]\leq \alpha,\ z[y_i]\geq -\alpha\}.
    \end{align*}
    Each of these orthants will be added with multiplicity $M$, defined below. See \Cref{fig:orthants1} for an illustration of the choice orthants for a single vertex $v_{i,\alpha}$.
    \begin{figure}[pos=t]
        \centering
        \includegraphics[width=0.45\linewidth]{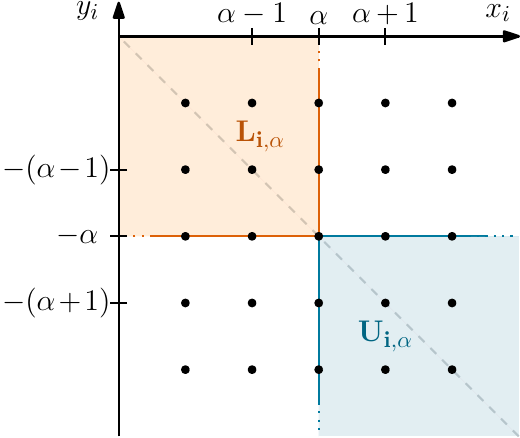}
        \caption{Illustration of the choice orthants $L_{i,\alpha}$ and $U_{i,\alpha}$ for a single color class $V_i$. The valid choices are the points $z$ with $z[x_i]=-z[y_i]\land z[x_i]\in [n]$ on the diagonal line (dashed gray).}
        \label{fig:orthants1}
    \end{figure}

    For every edge $v_{i,\alpha}v_{j,\beta}\in E_{ij}$ with $i,j\in[k], i<j$, introduce four mixed 2-sided \textit{edge orthants}:
    \begin{align*}
        &D_1=\{z:z[x_i]\leq \alpha,\ z[x_j]\geq \beta\},&&
        D_2=\{z:z[y_i]\leq -\alpha,\ z[y_j]\geq -\beta\},\\
        &D_1'=\{z:z[x_i]\leq \alpha-1,\ z[y_j]\geq -(\beta-1)\},&&
        D_2'=\{z:z[y_i]\leq -(\alpha+1),\ z[x_j]\geq \beta+1\}.
    \end{align*}
    See \Cref{fig:orthants2} for an illustration of the four orthants associated with an edge $v_{i,\alpha}v_{j,\beta}$.
    \begin{figure}[pos=t]
        \centering
        \includegraphics[width=0.45\linewidth]{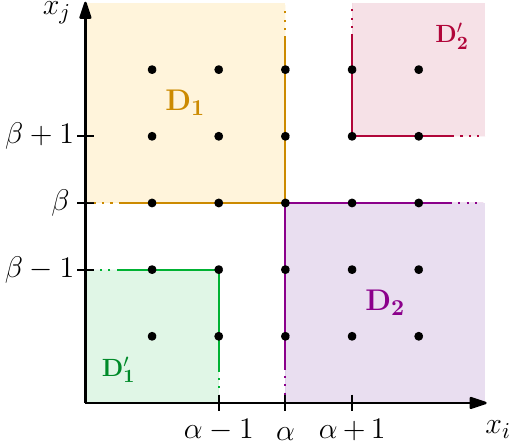}
        \caption{Illustration of the edge orthants associated with an edge $v_{i,\alpha}v_{j,\beta}$. Orthants are shown along the dimensions $x_i$ and $x_j$; assuming only valid choices, we represent a bound along $y_i$ (resp. $y_j$) via the inverted bound on $x_i$ (resp. $x_j$).
        }
        \label{fig:orthants2}
    \end{figure}

    The total number of edge orthants is $4\lvert E_H\rvert$, and we set $M:=4\lvert E_H\rvert+1$.
    Finally, we set the target to $\ell=kM(n+1)+\lvert E_H\rvert+\binom{k}{2}$. This completes the construction.

    {Consider any point $z\in \mathbb R^d$.} Among the choice orthants $U_{i,1},\ldots,U_{i,n}$ for some $i\in [k]$, precisely the ones with index at most ${\min\{z[x_i],-z[y_i]\}}$ contain $z$. Analogously, among $L_{i,1},\ldots,L_{i,n}$, precisely those with index at least ${\max\{z[x_i],-z[y_i]\}}$ contain $z$. Hence, $z$ is contained in at most $n+1$ of these $2n$ orthants. This maximum of $n+1$ orthants is reached by $z$ if and only if ${z[x_i]=-z[y_i]=\alpha_i}$ for some $\alpha_i\in[n]$, that is, if $z$ encodes a valid choice for color $i$.

    Now fix an edge $v_{i,\alpha}v_{j,\beta}\in E_{ij}$ and let $z\in\mathbb R^d$ be a valid choice, selecting $v_{i,{s}}\in V_i$ and $v_{j,{t}}\in V_j$. The four edge orthants $D_1,D_2,D_1',D_2'$ associated with $v_{i,\alpha}v_{j,\beta}$ contain $z$ under the following conditions:
    \begingroup\interdisplaylinepenalty=10000
    \begin{align*}
        z\in D_1 &\quad\Longleftrightarrow\quad {s}\leq\alpha \text{ and } {t}\geq\beta, &
        z\in D_2 &\quad\Longleftrightarrow\quad {s}\geq\alpha \text{ and } {t}\leq\beta,\\
        z\in D_1' &\quad\Longleftrightarrow\quad {s}<\alpha \text{ and } {t}<\beta, &
        z\in D_2' &\quad\Longleftrightarrow\quad {s}>\alpha \text{ and } {t}>\beta.
    \end{align*}\endgroup
    Thus, $D_1,D_2,D_1',D_2'$ contain every pair $({s},{t})\in[n]^2$ exactly once, except $(\alpha,\beta)$, which is contained in both $D_1$ and $D_2$.
    Consequently, the four edge orthants associated with $v_{i,\alpha}v_{j,\beta}$ contribute $2$ if ${s}=\alpha$ and ${t}=\beta$, and $1$ otherwise, to the depth of $z$.
    It follows that the edge orthants associated with $E_{ij}$ contribute $\lvert E_{ij}\rvert+1$ if $v_{i,{s}}v_{j,{t}}\in E_{ij}$, and $\lvert E_{ij}\rvert$ otherwise.

    It remains to prove that $H$ has a multicolored clique if and only if there is a point contained in at least $\ell$ orthants.

    \bigparagraph{$(\Rightarrow)$}\quad
    Assume $\{v_{1,\alpha_1},\ldots,v_{k,\alpha_k}\}$ is a multicolored clique in $H$, and consider the valid choice $z$ with \[
    \forall i\in [k]:z[x_i]=\alpha_i\land z[y_i]=-\alpha_i
    \]
    Then $z$ is contained in exactly $kM(n+1)$ choice orthants. Moreover, every pair $i<j$ contributes $\lvert E_{ij}\rvert+1$ edge orthants, since $v_{i,\alpha_i}v_{j,\alpha_j}\in E_H$. Hence, the total depth of $z$ is $kM(n+1)+\sum_{i<j}(\lvert E_{ij}\rvert+1)=kM(n+1)+\lvert E_H\rvert+\binom{k}{2}=\ell$.

    \bigparagraph{$(\Leftarrow)$}\quad
    Assume there is a point $z\in\mathbb R^d$ with depth at least $\ell$. 
    If $z$ does not encode a valid choice for some color $i\in[k]$, then the choice orthants can contribute at most $kM(n+1)-M$ in total. Since all edge orthants together contribute at most $4\lvert E_H\rvert$, the total depth of $z$ is at most $kM(n+1)-M+4\lvert E_H\rvert=kM(n+1)-1<\ell$, which is a contradiction.
    Hence $z$ must be a valid choice, selecting $v_{i,\alpha_i}\in V_i$ for $i\in[k]$.
    By the above observation, its depth equals $kM(n+1)+\lvert E_H\rvert+\sum_{i<j}\mathds{1}_{v_{i,\alpha_i}v_{j,\alpha_j}\in E_H}$.
    Reaching $\ell$ forces every one of the $\binom{k}{2}$ selected vertex pairs to be adjacent. Thus, the selected vertices form a multicolored clique in $H$.

    The construction has dimension $d=2k$ and produces exactly $N=2kn(4|E_H|+1)+4|E_H|$ orthants. Every orthant uses two distinct coordinates, one lower bound and one upper bound, where all bounds use integers in $[-(n+1),n+1]$. Thus the construction has polynomial size. Since $d=2k$ and \MCclique is \Wone-hard by $k$, this proves \Wone-hardness of \textup{\textsc{M2SOD}} by $d$.
\end{proof}

\section*{Acknowledgments and funding}
Michelle D\"oring and Georg Tennigkeit are supported by the HPI Research School on Foundations of AI (FAI).

\section*{Declaration of competing interest}
The authors declare that they have no known competing financial interests or personal relationships that could have appeared to influence the work reported in this paper.

\section*{Data availability}
No datasets were generated or analyzed in this theoretical study.

\section*{Declaration of generative AI and AI-assisted technologies in the manuscript preparation process}
During the preparation of this manuscript, the authors used OpenAI's Codex to assist with LaTeX formatting, bibliography integration, editorial suggestions, and drafting clarifications of mathematical definitions and reduction-size bounds. 

The proof idea was developed by the authors without the use of AI tools and the paper was also written without AI assistance (apart from the AI-assisted text formatting). 
The authors take full responsibility for the content of the manuscript.

\bibliographystyle{cas-model2-names}
\bibliography{MyLibrary}

@article{chan_fasteralgorithms_2023,
  title = {Faster Algorithms for Largest Empty Rectangles and Boxes},
  author = {Chan, Timothy M.},
  year = 2023,
  month = sep,
  journal = {Discrete \& Computational Geometry},
  volume = {70},
  number = {2},
  pages = {355--375},
  issn = {0179-5376, 1432-0444},
  doi = {10.1007/s00454-022-00473-x},
  urldate = {2026-09-22},
  langid = {english}
}

@inproceedings{chan_kleesmeasure_2013,
  title = {Klee's Measure Problem Made Easy},
  booktitle = {2013 {{IEEE}} 54th {{Annual Symposium}} on {{Foundations}} of {{Computer Science}}},
  author = {Chan, Timothy M.},
  year = 2013,
  month = oct,
  pages = {410--419},
  publisher = {IEEE},
  address = {Berkeley, CA, USA},
  doi = {10.1109/FOCS.2013.51},
  urldate = {2026-09-22},
  isbn = {978-0-7695-5135-7},
  langid = {english}
}

@article{fellows_parameterizedcomplexity_2009,
  title = {On the Parameterized Complexity of Multiple-Interval Graph Problems},
  author = {Fellows, Michael R. and Hermelin, Danny and Rosamond, Frances and Vialette, St{\'e}phane},
  year = 2009,
  month = jan,
  journal = {Theoretical Computer Science},
  volume = {410},
  number = {1},
  pages = {53--61},
  issn = {0304-3975},
  doi = {10.1016/j.tcs.2008.09.065},
  langid = {english}
}

@article{gorbachev_combinatorialdesigns_2023,
  title = {Combinatorial Designs Meet Hypercliques: Higher Lower Bounds for Klee's Measure Problem and Related Problems in Dimensions d {$\geq$} 4},
  shorttitle = {Combinatorial Designs Meet Hypercliques},
  author = {Gorbachev, Egor and K{\"u}nnemann, Marvin},
  editor = {Chambers, Erin W. and Gudmundsson, Joachim},
  year = 2023,
  journal = {LIPIcs, Volume 258, SoCG 2023},
  volume = {258},
  pages = {36:1-36:14},
  publisher = {Schloss Dagstuhl -- Leibniz-Zentrum f\"ur Informatik},
  issn = {1868-8969},
  doi = {10.4230/LIPICS.SOCG.2023.36},
  urldate = {2026-09-30},
  copyright = {Creative Commons Attribution 4.0 International license, info:eu-repo/semantics/openAccess},
  isbn = {9783959772730},
  langid = {english}
}
\end{document}